\documentclass[journal,comsoc]{IEEEtran}
\IEEEoverridecommandlockouts

\usepackage{amssymb}
\usepackage{bbm}
\usepackage{amsmath}
\usepackage{amsthm,amssymb}
\usepackage{caption}
\usepackage{cite}
\usepackage{mathrsfs}
\usepackage[displaymath,mathlines]{lineno}
\usepackage{color}
\usepackage{overpic}
\usepackage{pbox}
\usepackage{multicol}
\usepackage{amsthm}
\usepackage{epstopdf}

\usepackage[linesnumbered, ruled,vlined]{algorithm2e}
\usepackage{algorithmic}
\usepackage{soul}

\makeatletter
\newcommand{\doublewidetilde}[1]{{%
		\mathpalette\double@widetilde{#1}%
}}
\newcommand{\double@widetilde}[2]{%
	\sbox\z@{$\m@th#1\widetilde{#2}$}%
	\ht\z@=.5\ht\z@
	\widetilde{\box\z@}%
}
\makeatother

\newtheorem{lemma}{Lemma}

\begin{document}
\title{\huge Phase Shift Design for RIS-Aided Cell-Free Massive MIMO with Improved Differential Evolution \vspace{-0.2cm}}
%\author{Author 1, Author 2, Author 3,  Author 4, Author 5, and Author 6}
\author{Trinh~Van~Chien, \textit{Member}, \textit{IEEE}, Cuong V. Le, Huynh~Thi~Thanh~Binh, \textit{Member}, \textit{IEEE}, Hien~Quoc~Ngo, \textit{Senior Member}, \textit{IEEE}, and  Symeon Chatzinotas, \textit{Fellow}, \textit{IEEE} \vspace{-1cm}
\thanks{T. V. Chien, Cuong V. L., and H. T. T. Binh are with the School of Information and Communication Technology (SoICT), Hanoi University of Science and Technology (HUST), Hanoi 100 000, Vietnam (e-mail: chientv@soict.hust.edu.vn, cuonglv.hust@gmail.com, binhht@soict.hust.edu.vn). H. Q. Ngo is with the  Institute of Electronics, Communications and Information Technology, Queen's University Belfast, Belfast, U.K (email: hien.ngo@qub.ac.uk). S. Chatzinotas is with the University of
Luxembourg (SnT), Luxembourg (e-mail: symeon.chatzinotas@uni.lu). This research is funded by
Hanoi University of Science and Technology (HUST) under project number T2022-TT-001. The work of H. Q. Ngo was supported by the U.K. Research and Innovation Future Leaders Fellowships under Grant MR/X010635/1. The work of Symeon Chatzinotas was supported by the Luxembourg National Fund (FNR)-RISOTTI–the Reconfigurable Intelligent Surfaces for Smart Cities under Project FNR/C20/IS/14773976/RISOTTI.
}}
\maketitle

\begin{abstract}
This paper proposes a novel phase shift design for cell-free massive multiple-input and multiple-output (MIMO) systems assisted by reconfigurable intelligent surface (RIS), which only utilizes channel statistics to achieve the uplink sum ergodic throughput maximization under spatial channel correlations. Due to the non-convexity and the scale of the derived optimization problem, we develop an improved version of the differential evolution (DE) algorithm. The proposed scheme is capable of providing high-quality solutions within reasonable computing time. 
% Specifically, the proposed algorithm relies on biologically-inspired operators toward better solutions with the given measure of quality under the two mutation strategies.
Numerical results demonstrate  superior improvements of the proposed phase shift designs over the other benchmarks, particularly in scenarios where direct links are highly probable.
\end{abstract}
\vspace{-0.2cm}
\IEEEpeerreviewmaketitle
\vspace*{-0.2cm}
\section{Introduction}
%\vspace*{-0.15cm}
The next generation wireless systems are expected to provide very high connectivity for an extensive multitude of mobile devices. This poses major theoretical and practical challenges that require significant research beyond the state of the art. Cell-free massive multiple-input multiple-output (MIMO) has currently been considered as an emerging technology
to fulfil this requirement, for its ability to provide high macro-diversity and huge (virtual) array gain \cite{Ngo2017a}. However, in many practical scenarios, even with cell-free massive MIMO technology, some users may not receive a good quality of service due to high path loss with large obstacles and/or unfavourable scattering environments. One of the promising solutions to deal with the above  harsh propagation conditions is integrating reconfigurable intelligent surface (RIS) and cell-free massive MIMO. RIS is an effective energy-saving solution for enhancing wireless communication systems by carefully designing the phase shifts to obtain constructive combinations at receivers \cite{zhang2022active}. Thus, RIS-aided cell-free massive MIMO has received a lot of research interest recently \cite{9665300,al2021ris,van2021ris}.

Evolutionary algorithms (EAs) have gained significant attention due to their efficiency and scalability in solving real-world optimization problems. Owing to the complicated structure of future networks encompassing multiple integrated technologies, EAs such as the genetic algorithm (GA) have recently applied for resource allocation in 6G communications. This resoure allocation is based on either full channel state information \cite{9536231} or  channel statistics \cite{dai2022two} for the slow fading channel models. Among the class of EA algorithms, the differential evolution (DE) is one of the most powerful solvers to deal with numerical optimization problems \cite{das2016recent}. Similar to other EAs, DE initiates with a random population of individuals where each of them encodes for one solution. In next steps (which are also called generations), new solutions are produced from the current population using evolutionary operators, including crossover and mutation. Under the selection pressure, superior solutions are inclined to survive and impart their information to the subsequent generation. Through this iterative mechanism, solutions undergo refinement, eventually converge to a sub-optimal solution.
Due to its effectiveness and a compact structure, DE has been successfully applied for a wide range of real and complex optimization problems, such as engineering design, machine learning, data mining, planning and control \cite{das2016recent, 9743417}. However, the canonical DE algorithm still remains several limitations. Firstly, the standard DE uses only one mutation operator throughout the search process. Nevertheless, it is well-known that the performance of an evolutionary operator not only depends on the characteristics of the problem, but also the population's status. This means that even when chosen carefully, a mutation operator is only suitable at a certain stages of the evolution due to the changing population. Secondly, the DE performance  is strongly impacted by  control parameters, such as the crossover rate and the scale factor in the mutation operator. Since configuration of these parameters is problem-dependent, the algorithm necessitates careful tuning when applied to a specific real-world problem, which poses a significant issue in practice.   

In this paper, we demonstrate the possibility and effectiveness of the  phase shift design for cell-free massive MIMO with the support of an RIS to enhance the spectral efficiency of the uplink data transmission. Due to the non-convexity and the scale of the derived optimization problem, we propose an improved version of DE without suffering the previous mentioned limitations to find a sub-optimal solution. To the best of our knowledge, this is the first study exploiting the advantages of the DE  to address the long-term phase shift design for RIS-aided cell-free massive MIMO systems subject to spatial correlation between the scattering elements under the fast fading channel models, and practical conditions including imperfect channels and pilot contamination.\footnote{For slow-fading models, we may have sufficient time and radio resources to acquire highly accurate channel estimates. In contrast, in fast-fading scenarios, where radio resources are limited and coherence time is short, the estimation errors cannot be disregarded. Thus, practical conditions such as imperfect channels and pilot contamination should be taken into account.}
Our main contributions  can be briefly summarized as follows:
%\begin{itemize}
 $(i)$ we formulate a sum ergodic throughput optimization for the uplink data transmission of RIS-aided cell-free massive MIMO systems that designs the phase shifts based on the statistical channel information and spatial correlation between the scattering elements; $(ii)$ we develop an improved version of DE to find an efficiently sub-optimal solution to the phase shift design in polynomial time; and $(iii)$ numerical results show that our phase shift designs improve the uplink sum ergodic throughput more than $20\%$ compared to the state-of-the-art baselines. The results also verify the effectiveness and superiority of our proposed algorithm compared to both the canonical DE and GA.
%\end{itemize}

\textit{Notation}: Upper and lower bold letters  denote matrices and vectors. A diagonal matrix created from the vector $\mathbf{x}$ is denoted by $\mathrm{diag} (\mathbf{x})$. The superscripts $(\cdot)^T$ and  $(\cdot)^{\ast}$ are the regular transpose and the complex conjugate. The notations $\mathcal{CN} (\cdot, \cdot)$ and $\mathcal{U}([a, b])$ denote the circularly symmetric Gaussian distribution and the uniform distribution in  $[a, b]$. The expectation and variance of a random variable are $\mathbb{E} \{ \cdot \}$ and $\mathsf{Var} \{ \cdot \}$. Finally, $\mathsf{Pr}(\cdot)$ is the probability of an event.
\vspace{-0.2cm}
\section{System Model and Uplink Ergodic Throughput}
\vspace{-0.1cm}
We consider an RIS-aided cell-free massive MIMO system where $M$ access points (APs) coherently serve $K$ users, all having a single antenna.  The system performance is enhanced by the assistance of an RIS  equipped with $N$ phase shift elements. Let us mathematically denote the phase shift matrix as $\pmb{\Phi} = \mathrm{diag}([e^{j \theta_1}, \ldots, e^{j \theta_N}])$, where $\theta_n \in [-\pi, \pi]$ is the phase applied to the $n$-th RIS element. The channel between AP~$m$ and user~$k$ in the isotropic fading environment $g_{mk}$ is distributed as $g_{mk} \sim \mathcal{CN}(0,\beta_{mk})$, where $\beta_{mk}$ represents the large-scale fading effects. Each pair of cascaded channels from AP~$m$ to user~$k$ through the RIS consists of the two channels: $\mathbf{h}_m \sim \mathcal{CN}(\mathbf{0}, \mathbf{R}_m)$ from AP~$m$ to the RIS and $\mathbf{z}_k \sim \mathcal{CN}(\mathbf{0}, \widetilde{\mathbf{R}}_k)$ from the RIS to user~$k$.\footnote{Rayleigh fading channels are  particularly well-suited for rich scattering environments in sub-6GHz mobile communications.} Here, $\mathbf{R}_m$, and $\widetilde{\mathbf{R}}_k \in \mathbb{C}^{N \times N}$ are the corresponding spatial correlation matrices. The received uplink signal at the CPU is formulated as
\begin{equation} \label{eq:rk}
r_k = \sqrt{\rho} \sum\nolimits_{m=1}^M \sum\nolimits_{k'=1}^K   \hat{u}_{mk}^\ast u_{mk'} s_{k'} + \sum\nolimits_{m=1}^M \hat{u}_{mk}^\ast w_{m},
\end{equation}
where $\rho$ is the normalized uplink signal-to-noise ratio (SNR) of user $k$; $w_m$ denotes the additive white Gaussian noise with zero mean and unit variance; $u_{mk} = g_{mk} + \mathbf{h}_{m}^H \pmb{\Phi} \mathbf{z}_{k}$ is the aggregated channel between user~$k$ and AP~$m$ and its linear mean square error estimate (LMMSE) is denoted as $\hat{u}_{mk}$ \cite{9665300}. Maximum-ratio combining is exploited in \eqref{eq:rk} to detect the desired signals since this linear processing works well for single-antenna APs and can be easily implemented in a distributed manner \cite{Ngo2017a}. Through the utilization of the use-and-then-forget channel capacity bounding technique applied to \eqref{eq:rk} \cite{9665300}, we can obtain the uplink ergodic throughput of user $k$  as 
\begin{equation} \label{eq:ULRate}
	R_{k} (\pmb{\Phi}) = B \left( 1 - \tau_p/\tau_c \right) \log_2 \left( 1 + \mathrm{SINR}_{k} (\pmb{\Phi}) \right), \mbox{[Mbps]},
\end{equation}
where $B$ [MHz] is the system bandwidth, $\tau_c$ is the number of symbols in each coherence interval in which $\tau_p$ symbols are dedicated to the pilot training phase; and the signal-to-interference-and-noise ratio (SINR)  is 
\begin{equation} \label{eq:ClosedFormSINR}
	\mathrm{SINR}_{k} (\pmb{\Phi}) = \rho  \left( \sum\nolimits_{m=1}^M \gamma_{mk} \right)^2 / (\mathsf{MI}_{k} + \mathsf{NO}_{k}),
\end{equation}
where $\gamma_{mk}$ denotes the variance of the channel estimate, which is defined as $\gamma_{mk} = \mathbb{E} \{ |\hat{u}_{mk}|^2 \} = \sqrt{p\tau_p}  \delta_{mk}c_{mk}$ with $c_{mk} =  \sqrt{p\tau_p} \delta_{mk} /(p\tau_p \sum_{k' \in \mathcal{P}_k} \delta_{mk'} + 1)$, $\delta_{mk'} = \beta_{mk} + \mathrm{tr}(\pmb{\Theta}_{mk})$, $\pmb{\Theta}_{mk} = \pmb{\Phi}^H \mathbf{R}_m \pmb{\Phi} \widetilde{\mathbf{R}}_{k} $, and $\mathcal{P}_k$ is the pilot reuse index set. In \eqref{eq:ClosedFormSINR}, the mutual interference $\mathsf{MI}_{k}$ and the noise $\mathsf{NO}_{k}$ are 
\begin{align} \label{eq:MIuk}
	&\mathsf{MI}_{k} = \rho \sum\nolimits_{k'=1}^K \sum\nolimits_{m=1}^M  \gamma_{mk} \delta_{mk'} + p \tau_p \rho \times  \notag \\
	&\sum\nolimits_{k' =1 }^K \sum\nolimits_{k'' \in \mathcal{P}_k}    \sum\nolimits_{m=1}^M  \sum\nolimits_{m'=1}^M c_{mk}c_{m'k} \mathrm{tr}( \pmb{\Theta}_{mk'} \pmb{\Theta}_{m'k''}) \\
	&+  p \tau_p \rho \sum\nolimits_{k' \in  \mathcal{P}_k} \sum\nolimits_{m=1}^M   c_{mk}^2 \mathrm{tr}(\pmb{\Theta}_{mk'}^2) \notag \\
	& + p \tau_p \rho \sum\nolimits_{k' \in \mathcal{P}_k \setminus \{ k\} }  \left(\sum\nolimits_{m=1}^M c_{mk} \delta_{mk'} \right)^2,\notag\\
	& \mathsf{NO}_{k} =\sum\nolimits_{m=1}^M \gamma_{mk}, 
\end{align}
which demonstrate that the ergodic throughput in \eqref{eq:ULRate} depends on the various factors of the system model, phase shift design, and propagation environment such as the near-far effects, spatial correlation, and channel estimation quality. 
\vspace*{-0.2cm}
\section{Phase Shift Design for Uplink Sum Ergodic Data Throughput Maximization}
In this section, we  formulate and solve the phase shift design problem that maximizes the sum ergodic throughput. %We then present our proposed algorithm for solving the problem.%\footnote{Power control can further improve the system perform. However, to see the benefits of phase shift designs, in this work, we focus only on  phase shift optimization. The joint power and phase shift design is left for future work.}
\vspace*{-0.2cm}
\subsection{Problem Formulation}
The sum ergodic throughput for the uplink data transmission is formulated as
\begin{equation} \label{Prob:ULSumRate}
	\begin{aligned}
		& \underset{ \substack{\{ \theta_n \} } }{\mathrm{maximize}}
		& & f(\pmb{\Phi}) = \sum\nolimits_{ k=1 }^{K} w_k R_{k} (\pmb{\Phi}) \\
		& \,\,\textrm{subject to}
		& &  - \pi \leq \theta_n \leq \pi , \forall n =1 ,\ldots, N, 
	\end{aligned}
\end{equation}
where $w_k \geq 0$ is the weight that models the priority of user~$k$. Problem~\eqref{Prob:ULSumRate} should be applied for fast fading environments where the ergodic throughput is a relevant measurement metric by averaging over many realizations of small-scale fading coefficients. 
The phase shift design obtained by solving \eqref{Prob:ULSumRate} reduces the network planning cost since it can be utilized for multiple coherent intervals whenever the channel statistics unchanged. It belongs to the category of the long-term phase shift designs with our aim is to maximize the sum ergodic throughput relying on the statistical channel information comprising the large-scale fading coefficients and the spatial correlation among passive scattering elements of the RIS. We stress that, different from the short-term phase shift design in  previous works, our solution is of particular interest in practice since the solution can be deployed at least over many coherence intervals in which the channel statistics are unchanged. Due to the non-convexity and the complex expression of the SINR in \eqref{eq:ClosedFormSINR}, the global optimum to problem~\eqref{Prob:ULSumRate} is nontrivial to obtain. %\textcolor{blue}{We hereafter ex the two evolutionary algorithms with low computational complexity.}
\vspace*{-0.2cm}
\subsection{Phase Shift Design with Improved Differential Evolution}
%\textcolor{deeppink}{Does this line make you feel good :)) ??}
We present the improved DE algorithm for solving the phase shift design problem in \eqref{Prob:ULSumRate}. The main flow of the algorithm is kept as in the standard DE \cite{das2016recent}, outlined in Fig. \ref{fig:de}. In particular, the algorithm starts by initializing a population of individuals and maintains it during the search process. In each generation (or a main loop in the figure), the mutation and crossover operators are performed on every individual to create new offspring. Each generated offspring is then evaluated and compared directly to its parent, and whose that yields a higher ergodic throughput value will be selected for the next generation. 
% Fig. \ref{fig:de} outlines the basic steps of DE algorithm to obtain a sub-optimal solution of problem~\eqref{Prob:ULSumRate}. 
However, compared to the canonical DE, our improved version possesses two additional features. Firstly, instead of using only one mutation strategy, we combine two different operators, each with its own advantages can complement the other. Secondly, the control parameters in the mutation and crossover operators are dynamically adapted according to the search behaviour instead of fixing values as in the standard DE. Detailed descriptions of the proposed algorithm will be provided in the next subsections.

\subsubsection{Solution representation}
The population $\mathcal{Q}$ consists of $I$ individuals, where the $i$th individual, i.e., $i \in \{1, \ldots, N \}$, is a  $N$-dimensional vector $\pmb{\theta}_{i} = \{\theta_{i1}, \theta_{i2}, \ldots, \theta_{iN}\}$ of real numbers in the range $[ - \pi, \pi]$ that represents a possible solution to the phase shift matrix $\pmb{\Phi}$. At the beginning of the algorithm, all the individuals, i.e., the phase shift coefficients, are randomly initialized in the feasible domain as follows:
\begin{equation} \label{eq:ga_init}
	\theta_{in} = -\pi + 2 \pi \tilde{\alpha}_{in},\ \forall n=1, \ldots N,
\end{equation}
where $\tilde{\alpha}_{in} \sim \mathcal{U}([0, 1])$. In each generation, mutation is first performed on $\pmb{\theta}_p$ to create a mutant vector $\pmb{\theta}_v$, which is then  combined with its parent $\pmb{\theta}_p$ to form a new solution $\pmb{\theta}_o$. 
\begin{figure}[t]
    \centering
    \includegraphics[width=0.48\textwidth]{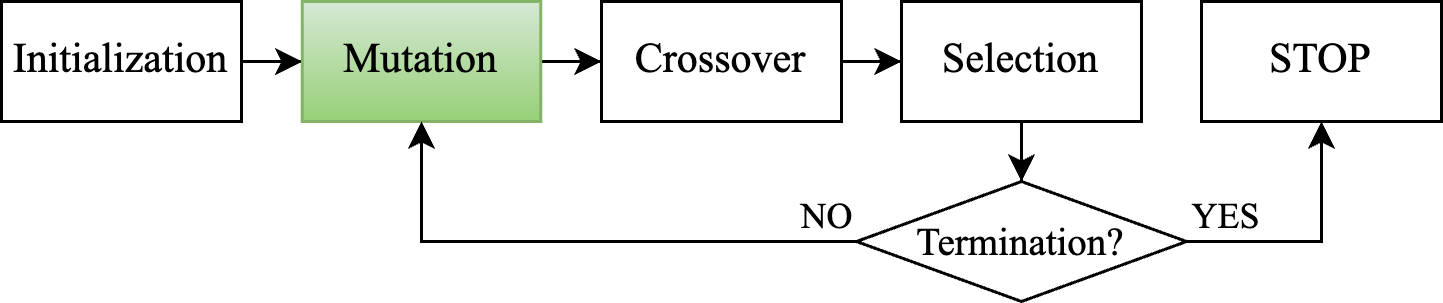}
    \caption{Basic steps of DE algorithm where the improvement is achieved for the mutation module.}
    \label{fig:de}
    \vspace{-0.7cm}
\end{figure}

\subsubsection{Mutation strategies}
Instead of using a single mutation operator, we employ two strategies with different characteristics, each suitable for specific problems or certain stages of evolution. Moreover, these mutation strategies are performed with different probabilities that are dynamically adjusted during the search according to their performance as: 
\begin{itemize}
    \item DE/$p$best/1 with the probability of $\lambda$:
    \begin{equation}
    \label{eq:pbest}
        \pmb{\theta}_{v} = \pmb{\theta}_{p\text{best}} + \mathsf{F} (\pmb{\theta}_{r_1} - \pmb{\theta}_{r_2}).
    \end{equation}
    
    \item DE/current-to-$p$best/1 with the probability of $(1 - \lambda)$:
    \begin{equation}
    \label{eq:ctpbest}
        \pmb{\theta}_{v} = \pmb{\theta}_p + \mathsf{F} (\pmb{\theta}_{p\mathrm{best}} - \pmb{\theta}_p) + \mathsf{F}(\pmb{\theta}_{r_1} - \pmb{\theta}_{r_2}),
    \end{equation}
\end{itemize}
where $\pmb{\theta}_{v}$ is the mutant vector corresponding to the parent solution $\pmb{\theta}_{p}$, $\pmb{\theta}_{p\text{best}}$ is selected randomly from the top  best solutions, $\pmb{\theta}_{r_1}$ and $\pmb{\theta}_{r_2}$ are selected randomly from the current population, and $\mathsf{F}$ is a scaled factor. The first operator, DE/$p$best/1 has been shown to provide fast convergence by combining the information of best solutions \cite{zhang2009jade}. Thus, this operator is suitable for unimodel problems or when the global basin was discovered in a multimodal problems. On the other hand, the random component $\mathsf{F}(\pmb{\theta}_{r_1} - \pmb{\theta}_{r_2})$ in the second operator results in global search behaviours, and therefore, this operator is appropriate at the beginning of the search when all the promising search regions need to be explored as soon as possible, especially in multimodal problems.

After every $G$ generations, the probability $\lambda$ is updated based on the effectiveness of the mutation strategies as follows:
\begin{equation}
\label{eq:lambda}
    \lambda = 
    \begin{cases}
        0.2, & \text{if } \frac{\Delta_1}{\mathrm{CFEs}_1} < \frac{\Delta_2}{\mathrm{CFEs}_2},\\
        0.8, & \text{otherwise,}
    \end{cases}
\end{equation}
where $\Delta_i \in \{ \Delta_1, \Delta_2 \}$ is the cumulative objective improvement gained by the $i-$th mutation strategy in previous $G$ generations and $\mathrm{CFEs}_i \in \{ \mathrm{CFEs}_1, \mathrm{CFEs}_2 \}$ is its number of consumed function evaluations (which indicates the computational resources that the $i$-th mutation strategy consumes in these previous $G$ generations). By this way, the effectiveness of mutation strategies is evaluated dynamically. The strategy with higher improvement rate is considered as more effective, and thus, is assigned higher probability and more computational resources.

\subsubsection{Crossover}
After the mutation, the parent vector $\pmb{\theta}_{p}$ is combined with the corresponding mutant vector $\pmb{\theta}_{v}$ using the binomial crossover to form the trial solution $\pmb{\theta}_{o} = \{\theta_{o1}, \theta_{o2}, \ldots, \theta_{oN}\}$, where $\theta_{on}, n= 1, \ldots, N,$ are
\begin{equation}
\label{eq:de_crossover} 
    \theta_{on} =
    \begin{cases}
        \theta_{vn}, & \text{if $\alpha_{on}$ $\leq \mathsf{CR}$ or $n = n_{\mathrm{rand}}$},\\
        \theta_{pn}, & \text{otherwise},
    \end{cases}
\end{equation}
where $\mathsf{CR}$ is a crossover rate; $\alpha_{on} \sim \mathcal{U}([0,1])$; and $n_{\mathrm{rand}}$ is an integer selected randomly from $[1, N]$ to ensure that $\pmb{\theta}_{o}$ gets at least one component from $\pmb{\theta}_{v}$.

\subsubsection{Survival selection} After generating the trial solution $\pmb{\theta}_{o}$ using the above mutation and crossover, the fitness value of $
\pmb{\theta}_{o}$ is calculated using the fitness function that is defined as same as the objective function of problem \eqref{Prob:ULSumRate}. The fitness value of $\pmb{\theta}_{o}$ is then compared directly to its parent $\pmb{\theta}_{p}$, and the better solution is admitted to the next generation:
\begin{equation}
    \label{de:selection}
    \pmb{\theta}'_{p} = 
    \begin{cases}
        \pmb{\theta}_{o}, & \text{if } f(\pmb{\Phi}_{o}) \geq f(\pmb{\Phi}_{p}),\\
        \pmb{\theta}_{p}, & \text{otherwise,}
    \end{cases}
\end{equation}
where the trial and parent versions of the phase shift matrix are defined as  $\pmb{\Phi}_o = \mathrm{diag}([e^{j \theta_{o1}}, \ldots, e^{j \theta_{oN}}]^T)$ and $\pmb{\Phi}_p = \mathrm{diag}([e^{j \theta_{p1}}, \ldots, e^{j \theta_{pN}}]^T)$;  $\pmb{\theta}'_{p}$ is the solution that will replace $\pmb{\theta}_{p}$ in the next generation. The condition in \eqref{de:selection} ensures the non-decreasing objective function of \eqref{Prob:ULSumRate} along the generations.

\subsubsection{Parameter adaptation}
The performance of the algorithm is strongly influenced by the parameters $\mathsf{F}$ and $\mathsf{CR}$ due to their roles in generating new solutions, as stated in \eqref{eq:pbest}, \eqref{eq:ctpbest}, and \eqref{eq:de_crossover}.
Instead of fixing these parameters as in the standard DE, we integrate a method called success-history based parameter adaptation (SHADE) \cite{tanabe2013success} into our algorithm for adapting $\mathsf{F}$ and $\mathsf{CR}$ automatically. Due to the space limitation, we omit the details of the SHADE method in this paper. In general, for the $i$-th mutation strategy, we use two memories called $\mathsf{MCR}_i$ and $\mathsf{MF}_i$, each of size $H$, to store the information of successful crossover rate and scale factor values, i.e., values that help to generate better solutions in previous generations. These stored successful values are then used to guide the algorithm to generate the crossover rate and the scale factor in the future generations. According to the descriptions above, we come to the skeleton of the proposal to the phase shift design given in Algorithm~\ref{pseu:de} together with its convergence property stated in the following lemma. 
\begin{lemma} \label{Theoremv1}
Let us define $O_\delta^\ast$ the space of the  $\delta$-optimal phase shift solution  to problem~\eqref{Prob:ULSumRate}, which is
\begin{equation}
\mathcal{S}_\delta^\ast = \left\{ \pmb{\theta}^\ast \big| |f(\pmb{\Phi}^\ast) - f(\pmb{\Phi})| \leq \delta,   -\pi \preceq \pmb{\theta}^\ast \preceq  \pi \right\},
\end{equation}
where $f(\pmb{\Phi})$ is defined in \eqref{Prob:ULSumRate} and $f(\pmb{\Phi}^\ast) = \max_{\pmb{\Phi}} \sum_{k=1}^K w_k R_k$. After that, for a population $\mathcal{Q}$ of $I$ initial individuals of the phase shift vector $\{\pmb{\theta}_i \}_{i=1}^I$ in the feasible domain, Algorithm~\ref{pseu:de} converges in probability to one solution $\pmb{\theta}^\ast \in \mathcal{S}_\delta^\ast$, i.e.,
\begin{equation} \label{eq:ProBound}
\mathsf{Pr}(\pmb{\theta}^\ast \in \mathcal{S}_{\delta}^{\ast}) \geq 1 - \left( 1 - \mu( S_{\delta}^\ast) P_{\mathrm{ep}}   \right)^I,
\end{equation}
where $P_{\mathrm{ep}} \in [0,1]$ is the mutation probability of each individual and $\mu( S_{\delta}^{\ast})$ is a measure to the space $S_{\delta}^{\ast}$.
\end{lemma}
\begin{proof}
The proof is to verify the existence of solution $\pmb{\theta}^\ast$ as Algorithm~\ref{pseu:de} improves the candidates along iterations. The detailed proof is available in the Appendix.
\end{proof}
\noindent Lemma~\ref{Theoremv1} offers two-fold: i) It confirms that each candidate of the phase shift vector enters the $\delta$-optimal solution space in probability; and ii) The convergence probability depends on the population size as shown in \eqref{eq:ProBound}. 

Regarding the computational complexity, the initialization step requires $\mathcal{O}(IN)$; in each generation, sorting the population to extract best solutions (used in Eq. \eqref{eq:pbest} and \eqref{eq:ctpbest}) requires $\mathcal{O}(I\log(I))$, mutation and crossover steps both require $\mathcal{O}(IN)$, the selection step requires $\mathcal{O}(I)$, and the parameter adaptation requires $\mathcal{O}(I)$ \cite{tanabe2013success}. Overall, the complexity of Algorithm \ref{pseu:de} is $\mathcal{O}(IN + G(I\log(I) + IN + IN + I + I))$ = $\mathcal{O}(GI\log(I) + GIN)$, where $G$ is the number of generations, $I$ is the population size, and $N$ is the number of phase shift elements.\footnote{Inspired by the maturity of the evolutionary algorithms, the improved DE-based phase shift design can be adapted to optimize the discrete phase shift coefficients. The adaptation holds particular interest for a future work.}

\begin{algorithm}[t]
    \caption{Improved DE-based phase shift design}
    \begin{algorithmic}[1]
        \label{pseu:de}
        \renewcommand{\algorithmicrequire}{\textbf{Input:}}
        \renewcommand{\algorithmicensure}{\textbf{Output:}}
        \REQUIRE Large-scale fading coefficients, spatial correlation matrices, channel estimation quality, bandwidth, and power coefficients. 
        \STATE Randomly initialize a population $\mathcal{Q}_P$ of $I$ individuals to the phase shift coefficients as in \eqref{eq:ga_init}.
        \STATE Calculate fitness for all individuals in $\mathcal{Q}_P$ using the objective function of problem~\eqref{Prob:ULSumRate}.
        \STATE Set the maximum number of generations $\mathrm{GEN}_{\mathrm{MAX}}$ and $\mathrm{GEN} \longleftarrow 0$.
        \WHILE{$\mathrm{GEN} < \mathrm{GEN}_{\mathrm{MAX}}$}
            \STATE Initialize the next generation population $\mathcal{Q}_{P'} \longleftarrow \O$;
            \FOR{each individual $\pmb{\theta}_{p}$ in $\mathcal{Q}_P$}
                \IF{ $\tilde{\alpha} \sim \mathcal{U}([0,1]) \leq \lambda$}
                    \STATE Generate mutant vector $\pmb{\theta}_{v}$ using  \eqref{eq:pbest}.
                \ELSE
                    \STATE Generate mutant vector $\pmb{\theta}_{v}$ using  \eqref{eq:ctpbest}.
                \ENDIF
                \STATE Generate the trial solution $\pmb{\theta}_{o}$ by combining $\pmb{\theta}_{p}$ and $\pmb{\theta}_{v}$ using  \eqref{eq:de_crossover}.
                \STATE Calculate fitness for for the trial solution $\pmb{\theta}_{o}$ using the objective function of problem~\eqref{Prob:ULSumRate}.
                \IF{$f(\pmb{\Phi}_{o}) \geq f(\pmb{\Phi}_{p})$}
                    \STATE $\mathcal{Q}_{P'} \longleftarrow \mathcal{Q}_{P'} \bigcup \{\pmb{\theta}_{o}\}$.
                \ELSE
                    \STATE $\mathcal{Q}_{P'} \longleftarrow \mathcal{Q}_{P'} \bigcup \{\pmb{\theta}_{p}\}$.
                \ENDIF
            \ENDFOR
            \STATE $\mathcal{Q}_{P} \longleftarrow \mathcal{Q}_{ P'}$.
            \FOR{each $i$-th mutation strategy}
                \STATE Update parameter memories $\mathsf{MF}_i$ and $\mathsf{MCR_i}$ as the SHADE method  \cite{tanabe2013success}.
            \ENDFOR
            \STATE Update the parameter $\lambda$ as in  \eqref{eq:lambda}.
            \STATE $\mathrm{GEN} \longleftarrow \mathrm{GEN} + 1$.
        \ENDWHILE
        \RETURN Best solution found.
    \end{algorithmic}
\end{algorithm}

% \begin{remark}
% \textcolor{blue}{The improved DE-based algorithm exploits mechanisms inspired by biological evolution to find the phase shift solution  maximizing the sum ergodic throughput in \eqref{Prob:ULSumRate} with a combination of the two distinguished mutation strategies in \eqref{eq:pbest} and \eqref{eq:ctpbest}. 
% This algorithm is strongly capable of identifying promising regions of the search space (exploration). Alternatively, the improved DE-based phase shift design generates a new solution by perturbing the current solution with the difference vectors between population members, and thus, it can overcome the limitation of the canonical DE thanks to the self-referential ability.\footnote{Inspired by the maturity of the evolutionary algorithms, the improved DE-based phase shift design can be adapted to optimize the discrete phase shift coefficients. The adaptation holds particular interest for a future work.} Some other evolutionary algorithms such as classical DE, partial swarm optimization (PSO) or genetic algorithm (GA) often fail or take a long time to refine the optimal local solution (exploitation).}
% \end{remark}
 \begin{figure*}[t]
\begin{minipage}{0.24\textwidth}
    \centering
    \includegraphics[trim=0.5cm 0cm 0.6cm 0.5cm, clip=true, width=1.8in]{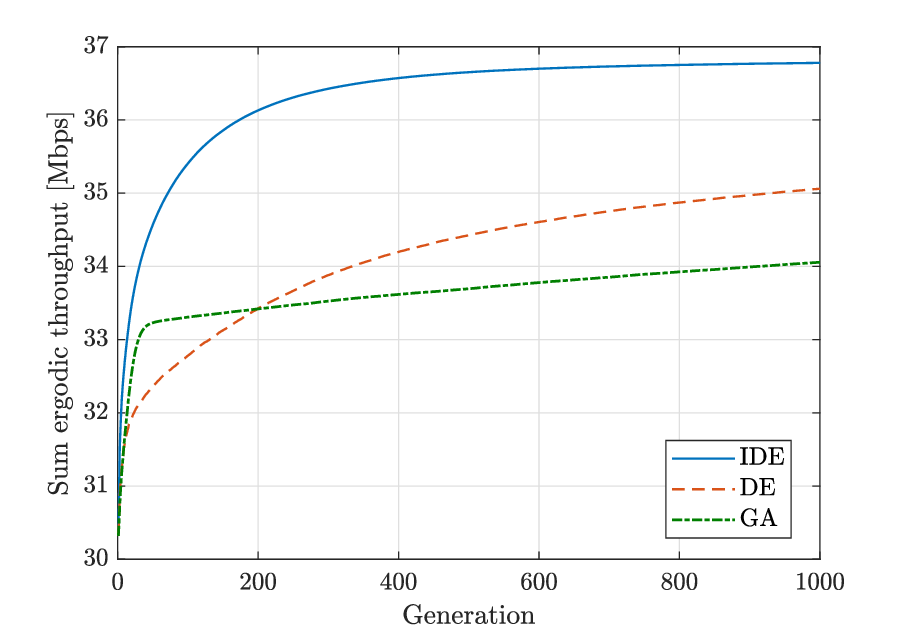} \\
    \fontsize{8}{8}{$(a)$}
    %\caption{Convergence of the proposed algorithms versus the generation index.}
    %\label{fig:Convergence100}
    \vspace{-0.2cm}
\end{minipage}
\hfil
\begin{minipage}{0.24\textwidth}
    \centering
    \includegraphics[trim=0.5cm 0cm 0.6cm 0.5cm, clip=true, width=1.8in]{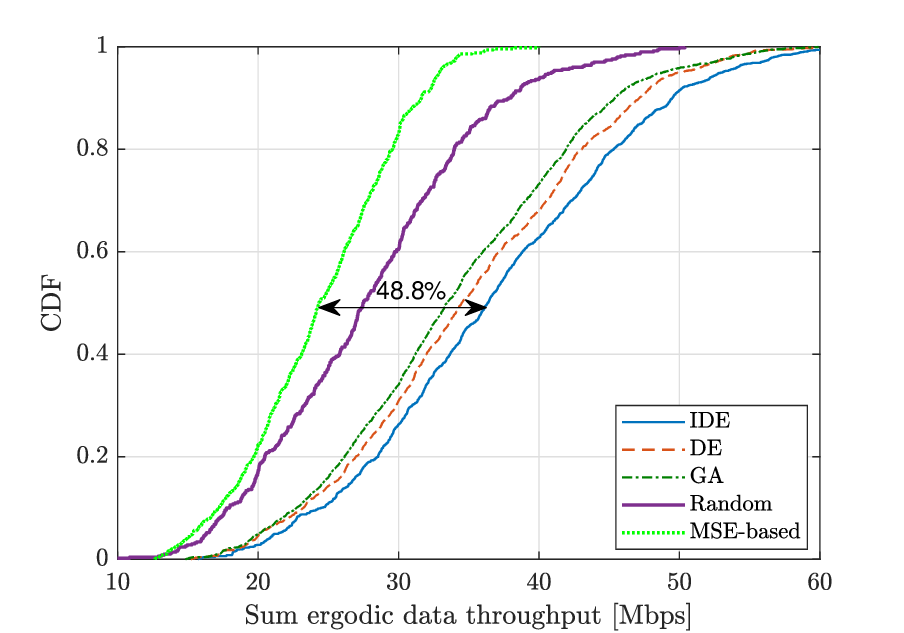}
    \fontsize{8}{8}{$(b)$}
    %\caption{CDF of the sum ergodic throughput [Mbps] with $N=100$.}
    %\label{fig:CDF100}
    \vspace{-0.2cm}
\end{minipage}
\hfil
\begin{minipage}{0.24\textwidth}
    \centering
    \includegraphics[trim=0.5cm 0cm 1.0cm 0.5cm, clip=true, width=1.8in]{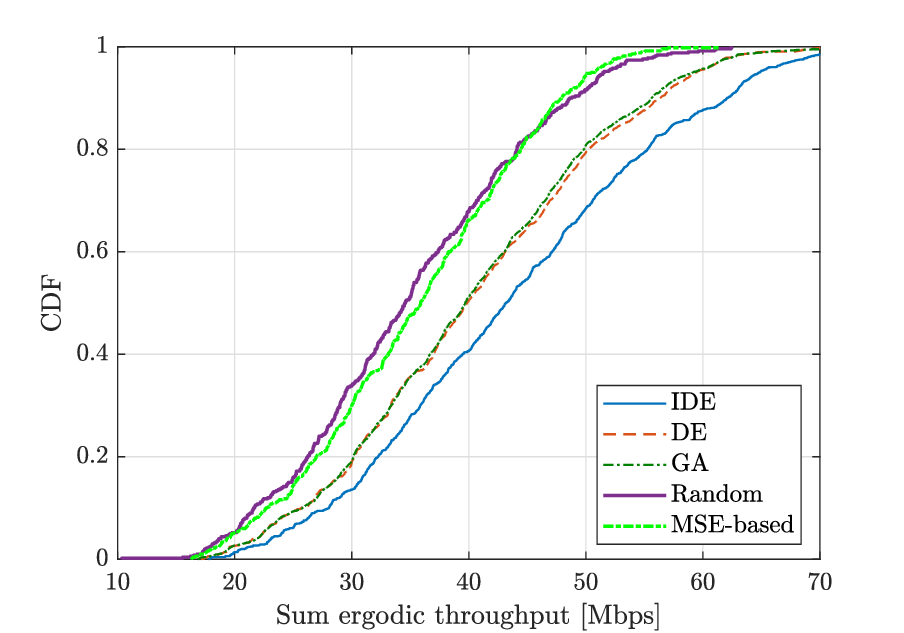}
    %\caption{CDF of the sum ergodic throughput [Mbps] with $N=256$.}
    %\label{fig:CDF256}
    \fontsize{8}{8}{$(c)$}
    \vspace{-0.2cm}
\end{minipage}
\hfill
\begin{minipage}{0.24\textwidth}
    \centering
    \includegraphics[trim=0.6cm 0cm 1.0cm 0.5cm, clip=true, width=1.8in]{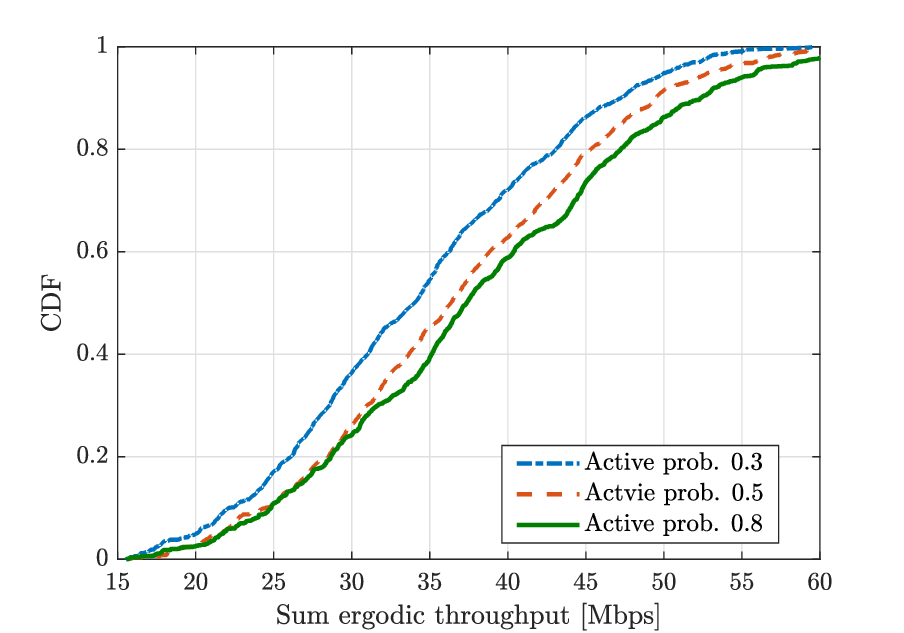}
    %\caption{\textcolor{blue}{CDF of the sum ergodic throughput [Mbps] with $N=100$ and different active  probability of direct links.}}
    %\label{fig:CDFActive}
    \fontsize{8}{8}{$(d)$}
    \vspace{-0.2cm}
\end{minipage}
\caption{The system performance of the benchmarks with different parameter settings: $(a)$ The convergence of the different evolutionary algorithms versus the generation index; $(b)$ The CDF of the sum ergodic throughput with $N=100$; $(c)$ The CDF of the sum ergodic throughput with $N=256$; and $(d)$ The CDF of sum ergodic throughput with the different active probabilities of direct links.}
    \label{fig2}
\vspace{-0.7cm}
\end{figure*}
\vspace*{-0.2cm}
\section{Numerical Results}
We consider an RIS-aided cell-free massive MIMO system comprising of $100$ APs serving $10$ users with a set of $5$ orthogonal pilot signals in the square area of $1$~km$^2$. The network topology is setup as in \cite{9665300} with the spatial correlation matrices defined by \cite{kammoun2020asymptotic}. The system bandwidth is $20$~MHz and the noise variance is $-92$~dBm. The direct links are unlocked with the probability $0.5$. The weights for the uplink sum ergodic are $w_k =1, \forall k$. 
Four benchmarks are involved for comparison: $i)$ \textit{Random phase shift design} (notated as Random) is widely used as a baseline in previous works \cite{wu2019intelligent}; $ii)$  \textit{Mean square error (MSE)-based phase shift design} (MSE-based) was proposed in \cite{9665300}, which obtains the global optimum as the direct links are totally blocked. However, this benchmark produces a sub-optimal solution under the presence of direct links with a non-neglectable probability; $iii)$ \textit{GA-based phase shift design} (GA) that exploits the genetic algorithm  \cite{harik1999compact}; $iv)$ \textit{DE-based phase shift design} (DE) that exploits the standard DE algorithm \cite{das2016recent}; and our \textit{Improved DE-based phase shift design} (IDE) is given in Algorithm~\ref{pseu:de}.

In Fig.~\ref{fig2}$(a)$, we plot the convergence trends of the employed evolutionary algorithms, including GA, DE, and our IDE. While all algorithms show significant improvements in the sum ergodic throughput throughout the evolution process, the IDE algorithm outperforms the others in terms of both convergence speed and quality of final solution. In comparison to the MSE-based baseline, the average improvement rate of the proposed IDE is about~$21.5\%$ for the network supported by an RIS equipped with $100$ scattering elements, and about $19.4\%$ if the RIS array gets bigger with $256$ scattering elements. %For comparison, we also include the random and MSE-based phase shift designs. Specifically, the evolutionary-based phase shift design algorithms outperform the other benchmarks with the gain up to $1.5\times$.   
For more details, we show the cumulative distribution function (CDF) of the sum ergodic throughput with the different number of scattering elements in Figs.~\ref{fig2}$(b)$ and \ref{fig2}$(c)$. Our proposed algorithms produce significantly better the spectral efficiency than the random and MSE-based phase shift designs under the presence of the direct links. The performance of MSE-based algorithm and the random phase shift design varies with the number of RIS elements since the number of optimization variables increases if the RIS is equipped with many scattering elements. It offers more degree of freedoms to obtain a good solution, and therefore the MSE-based algorithm can improve  the sum data throughput. %The MSE-based benchmark for an RIS-aided network with $100$ scattering elements is the baseline since the direct links appears quite frequently in our setting. However, it slightly performs better than the random phase shift design when the RIS is equipped with $256$ elements. 
Besides, the proposed algorithm provides the solution with the better sum ergodic throughput than the remaining evolutionary benchmarks. In particular, the gap between the improved DE-based and the other evolutionary phase shift designs becomes bigger as the number of scattering elements increase since a strategic mutation is required for a large RIS. The observation demonstrates the potentiality of the improved DE-based phase shift design for the large-scale systems. Finally, Fig.~\ref{fig2}$(d)$ plots the CDF of the sum ergodic throughput with different active probabilities of the direct links that demonstrate the contributions of the RIS in enhancing the spectral efficiency for harsh propagation environments.
\vspace*{-0.3cm}
\section{Conclusion}
\vspace*{-0.2cm}
This paper has manifested the benefits of the long-term phase shift design to improve the sum ergodic throughput of RIS-aided cell-free massive MIMO systems. The DE-based algorithm can  effectively handle the sophisticated nature of the sum throughput maximization with the presence of the RIS since they ideally do not rely on  the gradient of the objective function and constraints. In our considered settings, the long-term phase shift design obtained by the evolutionary algorithm produces nearly  $50\%$ higher sum ergodic throughput than the MSE-based solution at the median.
\vspace{-0.3cm}
\appendix \label{Appendix:Theoremv1}
By denoting $\mathcal{Q}(t) = \{ \pmb{\theta}_i (t) \}_{i=1}^I$ as the population that Algorithm~\ref{pseu:de} generates to solve problem~\eqref{Prob:ULSumRate} in which $t$ is the iteration index ($t=1,2, \cdots$), we recall the definition of the convergence in probability, which exists an $\pmb{\theta}_i(t)$ such that
%\begin{equation}
$ \lim_{t \rightarrow \infty} \mathsf{Pr} (\mathcal{Q}(t) \cap \mathcal{S}_\delta^\ast ) = 1$.
%\end{equation}
Consequently, by exploiting the same methodology as in  \cite[Theorem~1]{hu2013sufficient}, there exists a vector $\pmb{\theta}_i(t)$ that satisfies
%\begin{equation}
$\mathsf{Pr} (\pmb{\theta}_i(t) \in \mathcal{S}_\delta^\ast) \geq  1 - \eta(t_i)$,
%\end{equation}
where $\{ t_i | i = 1, 2, \ldots\}$ is a subsequence of the nature number set and $\eta (t_i)$ is a series such that $\sum_{i=1}^\infty \eta (t_i)$ diverges. Without loss of generality, one can select a positive number $\eta$ with $\eta(t_k) = \eta, \forall t_k$. Hence, the remaining activity is to define $\eta$. Indeed, let us formulate a measure to space $\mathcal{S}_\delta^\ast$ based on the mutation in \eqref{eq:pbest} and \eqref{eq:ctpbest} as
\begin{equation}
\mu(\mathcal{S}_\delta^\ast) = \mathsf{Pr}(\pmb{\theta}_i(t) \in \mathcal{S}_\delta^\ast) = \int_{\mathcal{S}_{\delta1}^\ast} g(\pmb{\theta}) d\pmb{\theta} +  \int_{\mathcal{S}_{\delta2}^\ast} h(\pmb{\theta}) d\pmb{\theta},
\end{equation}
where $g(\pmb{\theta})$ and $h(\pmb{\theta})$ are the probability density functions related to the randomness in \eqref{eq:pbest} and \eqref{eq:ctpbest}. The measures $\mathcal{S}_{\delta1}^\ast$ and $\mathcal{S}_{\delta2}^\ast$ are adopted to maintain $\pmb{\theta}_i(t) \in \mathcal{S}_\delta^\ast$. We can choose
%\begin{equation}
$\eta = 1 - \left( 1 - \mu( S_{\delta}^\ast) P_{\mathrm{ep}}   \right)^I$,
%\end{equation}
which demonstrates that $\eta \rightarrow 0$ as $I \rightarrow \infty$ since $P_{\mathrm{eq}}$ increases, the diversity of the population in Algorithm~\ref{pseu:de} will gradually improve. We therefore obtain the result as in the lemma.
\vspace{-0.3cm}
\bibliographystyle{IEEEtran}
\bibliography{IEEEabrv,refs}
\end{document}